\documentclass[journal]{IEEEtran}
\usepackage{mathtools}
\usepackage{subfigure}
\usepackage{bm}
\usepackage{graphicx}
\usepackage[utf8]{inputenc}
\usepackage{threeparttable}  
\usepackage{cite}
\usepackage{soul}
\usepackage{amsmath,amssymb,amsfonts}
\usepackage{svg}
\usepackage{algorithmic}
\usepackage{textcomp}
\usepackage{amsthm}
\usepackage{mathrsfs}
\usepackage{algorithm}
\usepackage{color,xcolor}
\usepackage{booktabs}  
\usepackage{siunitx}

\theoremstyle{definition}

\newtheorem{remark}{Remark}
\newtheorem{theorem}{Theorem}
\newtheorem{lemma}{Lemma}
\newtheorem{prop}{Proposition}
\usepackage{multirow}
\makeatletter

\newcommand{\Rmnum}[1]{\expandafter\@slowromancap\romannumeral #1@}
\makeatother
\newlength{\cellwidth}
\usepackage{ifpdf}

\DeclareMathOperator*{\Dg}{\mathsf{diag}}

\usepackage[utf8]{inputenc}
\usepackage[english]{babel}

\ifCLASSINFOpdf
\else
\fi
\usepackage{array}
\begin{document}
\title{Distributed Link Removal Strategy for Networked Meta-Population Epidemics and Its Application to the Control of the COVID-19 Pandemic}

\author{{Fangzhou~Liu$^{1}$, Yuhong~Chen$^{1}$, Tong Liu$^{1}$, Zibo Zhou$^{1}$, Dong~Xue$^{2}$, and Martin~Buss$^{1}$}% <-this % stops a space
%\thanks{This work was supported in part by the Self-Planned Task of State Key Laboratory of Robotics and Systems of Harbin Institute of Technology (No. SKLRS201801A03).}
\thanks{$^{1}$F. Liu, Y. Chen, T. Liu, Z. Zhou, and M. Buss are with the Chair of Automatic Control Engineering (LSR), 
Department of Electrical and Computer Engineering, Technical University of Munich, 
Theresienstr. 90, 80333, Munich, Germany; \texttt{\{fangzhou.liu, yuhong.chen, tong.liu, ga84sih, mb\}@tum.de}}  
\thanks{$^{2}$D. Xue is with the Key Laboratory of Advanced Control and Optimization for Chemical Processes, East China University of Science and Technology, Shanghai 200237, China; \texttt{dong.xue@ecust.edu.cn}}}

\maketitle

\begin{abstract}
 In this paper, we investigate the distributed link removal strategy for networked meta-population epidemics. In particular, a deterministic networked susceptible-infected-recovered (SIR) model is considered to describe the epidemic evolving process. In order to curb the spread of epidemics, we present the spectrum-based optimization problem involving the Perron-Frobenius eigenvalue of the matrix constructed by the network topology and transition rates. A modified distributed link removal strategy is developed such that it can be applied to the SIR model with heterogeneous transition rates on weighted digraphs. The proposed approach is implemented to control the COVID-19 pandemic by using the reported infected and recovered data in each state of Germany. The numerical experiment shows that the infected percentage can be significantly reduced by using the distributed link removal strategy. 
\end{abstract}

\begin{IEEEkeywords}
 distributed link removal strategy, networked meta-population epidemics, COVID-19 pandemic
\end{IEEEkeywords}

\IEEEpeerreviewmaketitle

\section{Introduction} \label{sec::intro}
Various models have been proposed to mathematically characterize the spread of epidemics \cite{mei_arc_2017,nowzari_ics_2016}. Among others, the compartmental models, e.g., the susceptible-infected-susceptible (SIS) model and the susceptible-infected-recovered (SIR) model, play the fundamental role. One important class of the compartmental models are the scalar deterministic models, which can be referred to in the survey \cite{hethcote_siam_2000}. These models have been widely investigated and qualitatively characterize the macroscopic behavior of the dynamics of infectious diseases, for example, the COVID-19 pandemic \cite{shi_medrxiv_2020,wu_lancet_2020}. However, the drawback of the scalar models is that they are based on the hidden assumption that there exists a well-mixed population, i.e., individuals have the same chances to interact with each other. In fact, this assumption introduces not only the homogeneity in network structure but also in individual behaviors, which does not generally hold in the globalized world with close connection via, for instance, face-to-face social networks and traffic networks. Both of these heterogeneities, nonetheless, play significant roles in shaping the epidemic spreading process. This brings us to the network epidemic models, where the nodal dynamics are considered. There are two kinds of interpretations of the network epidemic models: (a) the disease spreads on a network where each node represents one individual and (b) the disease spreads on a network of interconnected sub-population (groups of population), i.e., meta-population. Clear, the meta-population interpretation provides an efficient and comprehensive way of depicting pandemics which breaks out would-wide and spreads rapidly in communities. Thus, in this paper, we investigate the control strategy for networked meta-population epidemics. %from the perspective of changing interconnection topology.

In general, from the perspective of network systems, control strategies for network epidemics are categorized into node manipulation and edge manipulation. Previous literatures mainly focus on node manipulation, especially solving the resource allocation problem \cite{liu_tcns_2020,liu_physa_2019,nowzari_tcns_2017} by interacting with transition rates. Although they manage to control the disease spreading process with respect to certain optimal criteria, how to explicitly implement the control signal, which is associated with the modeling of the impact of the resources, remains to be explored. Instead, the edge manipulation has been widely applied in the world to curb the epidemic spreading. For example, city lock-down can be regarded as cutting off all the connections in the graph. 
By using link-removal strategy, the spectral radius of the graph can be decreased such that it is below the epidemic threshold \cite{mieghem_ton_2009} resulting in the disease-free equilibrium of the network epidemic model. However, it is a combinatorial and NP-hard problem to optimally design the network topology on which the epidemic spreads \cite{mieghem_pre_2011}. Recently, Xue and Hirche propose an algorithm combining power iteration (PI), max-consensus, and event-trigger optimization to approximately solve this problem in a distributed manner \cite{xue_tsp_2019}. Their method requires the network to be undirected. Nonetheless, real epidemic spreading network is generally directed or at least bi-directed based on the fact of asymmetric traffic flow between locations. In addition, the epidemic threshold adopted in \cite{mieghem_pre_2011,xue_tsp_2019} is rooted in an epidemic model with homogeneous transition rates, i.e., the transition rates of each sub-population are identical. Thus new criterion based on epidemic models with heterogeneous transition rates \cite{liuj_tac_2019} needs to be introduced. 

The main contribution of this paper is to develop a distributed link removal strategy to control networked SIR meta-population epidemics. We extend the algorithm in \cite{xue_tsp_2019} such that it is applicable for meta-population epidemics with heterogeneous transition rates on weighted digraphs. This distributed strategy enjoys the advantage of locally retrievable information and no centralized decision-maker.
Besides, from practical point of view, we implement the proposed algorithm to curb the spread of COVID-19 pandemic by using real data. We identify the infection matrix (infection rates and adjacency matrix of the network) and the curing rates by using the reported infection and recovered cases of each state of Germany. Built upon the identified infection matrix and curing rates, we show that the infected percentage of Germany can be significantly reduced by applying the proposed algorithm.

\emph{Notations:} Let $\mathbb{R}$, $\mathbb{N}$, and $\mathbb{N}_{\geq 0}$ be the set of real numbers, nonnegative integers, and positive integers, respectively. 
Given a matrix $M \in \mathbb{R}^{n \times n}$, $\lambda_i(M)$ is the $i$th largest eigenvalue of $M$ sorted in the decreasing order $|\lambda_1(M)| \geq |\lambda_2(M)| \geq \ldots \geq |\lambda_n(M)|$, $\rho(M)$ is the spectral radius of $M$, i.e., $\rho(M) = \max_{i}|\lambda_i(M)|$. Let $\mathsf{Re}(\lambda)$ be the real part of the eigenvalue $\lambda$. $\alpha(M)$ denotes the largest real part of $M$'s eigenvalues, i.e., $\alpha(M) = \max_i \mathsf{Re}(\lambda_i(M))$. For a matrix $M \in \mathbb{R}^{n \times r}$ and 
a vector $a \in \mathbb{R}^n$, $M_{ij}$ and $a_{i}$ denote the element in the $i$th row and $j$th column and the $i$th entry, respectively. For any two vectors $a, b \in \mathbb{R}^n$, $a \gg (\ll) b$ represents that $a_i >(<) b_i$, for all $i=1,\ldots,n$; $a > (<) b$ means that $a_i \geq (\leq) b_i$, for all $i=1,\ldots,n$ and $a \neq b$; and $a \geq (\leq) b$ means that $a_i \geq (\leq) b_i$, for all $i=1,\ldots,n$ or $a = b$. These component-wise comparisons are also applicable for 
matrices with the same dimension. Vector $\mathbf{1}$ ($\mathbf{0}$) represents the column vector of all ones (zeros) with appropriate dimensions. $I_n$ stands for the identity matrix of order $n$ and $e_i$ is the $i$th column of $I_n$.

\section{Problem Formulation} \label{sec::prob_form}
In this section, we recall some necessary notions from graph theory, introduce the meta-population SIR model, and provide the problem formulation.

\subsection{Preliminaries} \label{subsec::preliminaries}
We consider a social network described by a weighted directed graph $\mathcal{G}(\mathcal{V},
\mathcal{E}, W)$ with $n$ $(n \geq 2)$ nodes, where $\mathcal{V}=\{1,2,\dots,n\}$ and $\mathcal{E} \subseteq \mathcal{V} \times \mathcal{V}$ are the sets of nodes and edges, respectively. 
The adjacency matrix $W=[w_{ij}] \in \mathbb{R}^{n \times n}$ is nonnegative and with 
zero diagonal entries. For two distinct nodes, $w_{ij} >0$ if and only if there exists 
a link from node $j$ to $i$, i.e., $(j,i) \in \mathcal{E}$. For the convenience of further presentation, the in-neighborhood of node 
$i$ is also introduced as
\begin{equation}
\mathcal{N}^{\text{in}}_i=\{j:w_{ij}>0, j \in \mathcal{V}\}.
\end{equation}
In this article, we confine ourselves that the graph $\mathcal{G}$ is strongly 
connected, i.e., $W$ is irreducible. We then introduce the Perron-Frobenius Theorem for irreducible nonnegative matrix.
\begin{lemma}{\cite[Theorem 2.7]{varga_springer_2000}} \label{lemma::perron_fr}
 Given that a square matrix $M$ is an irreducible nonnegative matrix. The following statements hold:
 \begin{itemize}
 	\item [(i)] The largest eigenvalue of $M$, $\lambda_1(M)$, is a positive real eigenvalue equal to its spectral radius $\rho(M)$.
 	\item [(ii)] $\rho(M)$ is a simple eigenvalue of $M$.
 	\item [(iii)] There exist a unique right eigenvector $y \gg \mathbf{0}$ and a unique left eigenvector $z^{\top} \gg \mathbf{0}^{\top}$ corresponding to $\rho(M)$.
 \end{itemize} 
\end{lemma}
Since the adjacency matrix $W$ of the strongly connected graph $\mathcal{G}$ is nonnegative irreducible, Lemma~\ref{lemma::perron_fr} can be directly applied. In the remaining of this article, we denote positive vectors $y=[y_1,y_2,\ldots,y_n]^{\top}$ and 
$z^{\top}=[z_1,z_2,\ldots,z_n]$ the right and left eigenvector corresponding to $\rho(W)$.

\subsection{Meta-Population Susceptible-Infected-Recovered Model}
Consider epidemics spreading on a weighted directed graph $\mathcal{G}(\mathcal{V},
\mathcal{E}, W)$. The dynamics of each group $i \in \mathcal{V}$ satisfies 
the meta-population SIR model as follows 
\begin{equation} \label{eq::sir}
 \begin{aligned}
  \dot{x}_i(t) &= (1-x_i(t)-r_i(t)) \sum_{j=1}^{n} \beta_{i} w_{ij} x_j(t) - \delta_i x_i(t) \\
  \dot{r}_i(t) &= \delta_i x_i(t),
 \end{aligned}
\end{equation}
 %$\beta_{ij} := \beta_i w_{ij}$ denotes the weighted infection rate via the edge $(i,j)$, where $\beta_i > 0$ is the infection rate of group $i$.
where $x_i(t), r_i(t) \in \mathbb{R}$ represent the proportions of infected (I) and recovered (R) cases in group $i$ at time instant $t$, respectively. $\beta_i,\delta_i > 0$ are the infection and curing rate of group $i$, respectively.  Since the transition rates are generally different for each group and the groups are not well-mixed, the model in~\eqref{eq::sir} characterizes the heterogeneity of the epidemic spreading process. From practical point of view, we adopted the SIR model described by $x_i$ and $r_i$, because the infected and recovered cases are regularly reported while the proportion of susceptible individuals in population can be hardly known. In addition, the dynamics of $i$th group's proportion of susceptible cases $s_i(t)$ can be omitted in the SIR model~\eqref{eq::sir} in light of the fact that 
$s_i(t) + x_i (t) + r_i(t) \equiv 1$ for all $i \in \mathcal{V}$ and $t \geq 0$. Furthermore, it is desirable that the states in the model~\eqref{eq::sir} stay in the following simplex.
\begin{equation}
    \Delta = \{(a,b) : a,b \geq 0, a+b \leq 1 \}.
\end{equation}

Let $x(t) = [x_1(t),x_2(t),\ldots,x_N(t)]^{\top}$ and $r(t) = [r_1(t),r_2(t),\ldots,r_N(t)]^{\top}$ be the stacked infection proportions and recovering proportions, respectively. Let $\beta = [\beta_1,\beta_2,\ldots,\beta_N]^{\top}$ and $\delta = [\delta_1,\delta_2,\ldots,\delta_N]^{\top}$. By denoting $X(t) = \Dg(x(t))$, $R(t) = \Dg(r(t))$, $B = \Dg(\beta)$, and $D = \Dg(\delta)$, the compact form of the meta-population SIR model~\eqref{eq::sir} reads
\begin{equation} \label{eq::sir_mat}
    \begin{aligned}
     \dot{x}(t) &= (1-X(t)-R(t)) BW x(t) -D x(t) \\
     \dot{r}(t) &= -D x(t).
    \end{aligned}
\end{equation}
The SIR model is considered in this article due to its wide application in the describing epidemic spreading process. Nonetheless, the results and algorithms in the following sections can be straightforwardly extended to other compartmental models, e.g., SI, SIS, SIRS, SEIR. 

In the scenarios of epidemic curbing and rumor mitigation, the disease-free case, i.e., $x=\mathbf{0}$, is of great significance. The following lemma collects the behavior of the 
meta-population SIR model~\eqref{eq::sir_mat}.
\begin{lemma} \label{lemma:sir}
 Consider the meta-population SIR model~\eqref{eq::sir_mat} with positive transition rates on a strongly connected weighted digraph $\mathcal{G}(\mathcal{V}, \mathcal{E}, W)$. The following statement hold:
 \begin{enumerate}
     \item[(i)] If $(x_i(0),r_i(0)) \in \Delta$ for all $i \in \mathcal{V}$, there holds $(x_i(t),r_i(t)) \in \Delta$ for all $i \in \mathcal{V}$ and $t \geq 0$.
     \item[(ii)] The set of equilibrium points is the set of pairs $(\mathbf{0}_N,r^*)$, for any $r^* \in [0,1]^n$.
     \item[(iii)] If $\alpha(BW-D) \leq -\epsilon$ for some $\epsilon > 0$, $x(t)$ approaches $\mathbf{0}$ exponentially fast, i.e., $\|x(t)\| \leq \|x(0)\|K e^{-\epsilon t}$, for some $K > 0$. 
 \end{enumerate}
\end{lemma}
\begin{proof}
 Consider the dynamics of each node in~\eqref{eq::sir}. Assume that for some time instant $\tau \geq 0$, there hold $(r_i(\tau),x_i(\tau)) \in \Delta$ for all $i \in \mathcal{V}$. We then inspect the following three cases: (a) if $x_i(\tau) = 0$, then $\dot{x}_i(\tau) \geq 0$; (b) if $r_i(\tau) = 0$, then 
 $\dot{r}_i(\tau) \geq 0$; and (c) if $x_i(\tau) + r_i(\tau) = 1$, then 
 $\dot{x}_i(\tau)+\dot{r}_i(\tau) = -\delta_i x_i(\tau) \leq 0$. By combining the above three cases, we can obtain the statement (i).
 
 The statement (ii) has been proved in \cite{mei_arc_2017} and the proof is saved for triviality.
 
 We then prove the statement (iii). In light of the statement (i), we can obtain
 \begin{equation}
    \dot{x} < (BW-D) x.
 \end{equation}
 Thus by comparison principle \cite{Khalil_prentice_2002}, we only need to prove the auxiliary system $\dot{y} = (BW-D) y$ converges to $\mathbf{0}$ exponentially fast. It is straightforwardly true since there hold $\alpha(BW-D) \leq -\epsilon$. Thus we complete the proof.
\end{proof}
\begin{remark}
By Lemma \ref{lemma:sir}, it is straightforward that $\Delta$ is an invariant set for the infection and recovered proportions, given nonnegative transition rates. Note that distinct from the SIS model, the SIR model always converge to a disease-free case if the curing rate $\delta_i$ is positive. In this regard, what matters for the SIR model is not whether there will be healthy state, but how fast the disease dies out. By the statement (iii), the decay rate of the infection proportion is furnished by $\alpha(BW-D)$. Thus, we control $\alpha(BW-D)$ to curb the spread of the epidemics.
\end{remark}

\subsection{Link Removal Problem}
Given a weighted digraph $\mathcal{G}=\{\mathcal{V},\mathcal{E},W\}$, the link removal problem is formally described as follows: for a fixed budget $| \Delta \mathcal{E} |=r (r \in \mathbb{N}_+)$, select a set of edges $\Delta \mathcal{E}$ from $\mathcal{\mathcal{E}}$ to construct a new graph $\mathcal{G}_r = \{\mathcal{V},\mathcal{E} \setminus \Delta \mathcal{E}, W_r \}$, such that the exponential decay rate of the meta-population SIR model~\eqref{eq::sir_mat} with positive infection and curing rates is maximized, i.e.,
\begin{equation} \label{eq::opt_1}
 \begin{aligned}
    \max_{\Delta \mathcal{E} \subseteq \mathcal{E}} \quad &\epsilon \\
    \mathrm{s.t.} \quad & \alpha(B W_r-D) \leq -\epsilon \\
                 \quad  & |\mathcal{E} |= r.
 \end{aligned}
\end{equation}
Note that since $B$ and $D$ are diagonal matrices with positive diagonal entries, the optimization problem~\eqref{eq::opt_1} is equivalent to minimize $\alpha(D^{-1}BW_r)$, where the matrix $D^{-1}BW_r$ is irreducible nonnegative. For the convenience of presentation, we denote 
\begin{equation} \label{eq::def_1}
   \begin{aligned}
    & \Delta W = W - W_r, \quad A := D^{-1}BW,  \\
    & A_r := D^{-1}BW_r, \quad \Delta A_r := D^{-1}B \Delta W.
   \end{aligned}
\end{equation}
In light of Lemma \ref{lemma::perron_fr}, we can rewrite the optimization problems as follows
\begin{equation} \label{eq::opt_2}
 \begin{aligned}
    \min_{\Delta \mathcal{E} \subseteq \mathcal{E}} \quad & \lambda_1(A_r) \\
    \mathrm{s.t.}  \quad  & |\mathcal{E} |= r.
 \end{aligned}
\end{equation}
After labeling edge $(j,i) \in \mathcal{E}$ on graph $\mathcal{G}$ by $l_{ij}$, the
optimization problem~\eqref{eq::opt_2} can be reformulated as
\begin{equation} \label{eq::opt_3}
 \begin{aligned}
    \min_{m \in \{0,1\}^{|\mathcal{E}|} } \quad & \lambda_1(A- \Delta A_r) \\
    \mathrm{s.t.}  \quad  & \Delta A_r = \sum_{l_{ij}=1}^{|\mathcal{E}|} \beta_i m_{l_{ij}} e_i e_j^{\top} w_{ij} / \delta_i \\
                   \quad & \mathbf{1}^{\top} m = r,
 \end{aligned}
\end{equation}
where $m = [m_1,m_2,\ldots,m_{|\mathcal{E}|}]^{\top}$ with $m_{l_{ij}} = 1$ if the edge labeled as $l_{ij}$ is removed from $\mathcal{E}$ and $m_{l_{ij}} = 0$, otherwise.

\begin{remark}
 For epidemics with infection rate $\beta$ and curing rate $\gamma$, we have the reproduction number $R = \frac{\beta}{\gamma}$. If $R<1$, the smaller $R$ is, the faster the epidemic dies out. For meta-population SIR model~\eqref{eq::sir_mat}, the dominant eigenvalue of the matrix $A$ can be considered as the reproduction number which takes into consideration the influence of the network topology as well as the heterogeneous transition rates. 
\end{remark}

\section{Distributed Link Removal Strategy}
In this section, we propose a distributed algorithm to solve the link removal problem~\eqref{eq::opt_3} for the meta-population SIR epidemic model on weighted digraphs. By using eigenvalue-sensitivity-based approximation, we introduce the dominant left and right eigenvectors of the matrix $A$ to solve the problem in question. Then we design a distributed algorithm based on power iteration and max-consensus algorithm.   

\subsection{Eigenvalue-Sensitivity-Based Approximation}
\begin{prop} \label{prop::approx}
 The optimization problem~\eqref{eq::opt_3} can be approximately solved by 
 \begin{equation} \label{eq::opt_final}
  \begin{aligned}
    \min_{m \in \{0,1\}^{|\mathcal{E}|} } \quad & \Delta \lambda_1(A, \Delta A_r) \\
    \mathrm{s.t.}  \quad  & \Delta \lambda_1(A, \Delta A_r) = \sum_{l_{ij}=1}^{|\mathcal{E}|} \beta_i m_{l_{ij}} w_{ij} z_i y_j / \delta_i \\
                   \quad & \mathbf{1}^{\top} m = r.
  \end{aligned} 
 \end{equation}
\end{prop}
\begin{proof}
 Since the graph $\mathcal{G}=\{\mathcal{V},\mathcal{E},W\}$ is strongly connected, $\lambda_1(A)$ is positive and simple by Lemma \ref{lemma::perron_fr}. In addition, the right and left (normalized) eigenvectors, $y$ and $z^{\top}$, are strictly positive, i.e., $y,z \gg \mathbf{0}$. According to the perturbation theory in \cite[p.183]{stewart_academic_1990}, the following expansion holds
 \begin{equation} \label{eq::derive_opt}
    \lambda_1(A-\Delta A_r) =\lambda_1(A) - \frac{z^{\top} \Delta A_r y}{z^{\top} y} + \mathcal{O}(\|\Delta A_r\|).
 \end{equation}
 For graphs $\mathcal{G}$ with a large spectral gap between $\lambda_1(W)$ and $\lambda_1(A)$, the higher order items can be neglected and the first-order approximation equals $\lambda_1(A-\Delta A_r)$. Since $y$ and $z^{\top}$ are normalized, i.e., $z^{\top}y=1$, we can obtain the expression of $\Delta \lambda_1(A, \Delta A_r)$ by rewritten $z^{\top} \Delta A_r y$ in a component-wise manner.
\end{proof}

\subsection{Distributed Algorithm Design}
To implement the algorithm in a distributed way, we firstly introduced distributed estimation of the eigenvectors. Based on the estimated eigenvectors, we carried out the removal algorithm. 
Power iteration (PI) is a common method to estimate dominant eigenvalue.The eigenvector  corresponding to the dominant eigenvalue of matrix $C$ is given by
\begin{equation} \label{eq::pi}
    \hat{\xi}(t+1)=\frac{C \hat{\xi}(t)}{\|C \hat{\xi}(t)\|},
 \end{equation}
where $\hat{\xi}(t)$ is estimation at step $t$. It worth noting that, power iteration demand a primitive matrix $C$, which may not be the case for defined matrix $A$. So we set $C=I+A$ to acquire the eigenvector. 

A problem for the distributed complement would be the normalization in \eqref{eq::pi} at each iteration step. To get $\|C \hat{\xi}(t)\|$, global information is needed. Therefore, we use a $\max$-consensus protocol to help getting the 
$\hat{\xi}(t)$ converged, not   normalized though.

For $\hat{y}(t)=[\hat{y}_1(t),...,\hat{y}_n(t)]^{\top}$, each node $i$ has access to its own value $\hat{y}_i(t)$ and its neighbours' value $\hat{y}_j(t)$, $j \in {N}_i$. Then the PI in \eqref{eq::pi} can be modified as
\begin{equation}\label{eq::eig_1}
    \hat{y}_i(t+1)=k_i(t) \left( \hat{y}_i(t) + \sum_{j=1}^{n} w_{ij} \hat{y}_j(t)\right),
\end{equation}
where $k_i(t)$ helps in the convergence of $\hat{y}_i(t+1)$ and can be achieved in the following $\max$-consensus way. Firstly, we calculate a candidate locally, as
\begin{equation}
 \begin{aligned}
   h_i(t+1) &=\frac{1}{\hat{y}_i(t)}\left[\hat{y}_i(t) + \sum_{j=1}^{n} w_{ij} \hat{y}_j(t) \right].
 \end{aligned}
\end{equation}
Then, every node shares this value with its neighbors and choose the max value from the values of its neighbors' and itself's:
\begin{equation}
 p_i(t+t_s+1) = \max_{j \in \mathcal{N}^{\mathrm{in}}_i} p_j(t + t_s), p_i(t) = h_i(t),
\end{equation}
where $t_s \in \mathbb{N}_{\geq 0}$. Terminated at a mixing time $T_d$, which means at time $t+T_d$ a $\max$-consensus is reached that
\begin{equation}
 p_1(t+T_d)=\cdots=p_n(t+T_d) = \max_{j\in\mathcal{V}} h_j(t).
\end{equation}
Setting $ k_i(t)=\frac{1}{p_i(t)}$, $ k_i(t)$ can be formulated as
\begin{equation}\label{eq::eig_2}
    k_i(t) =  \frac{1}{\max_{j\in\mathcal{V}} h_j(t-T_d)}.
\end{equation}
\begin{theorem}
Given a connected graph $\mathcal{G}$, the eigenvector $y$ corresponding to $\lambda_1(A(\mathcal{G}))$ can be computed distributively by repeating steps \eqref{eq::eig_1}-\eqref{eq::eig_2}.
\end{theorem}
\begin{proof}
According to \cite{wood2003always}, if irreducible $C\geq0$ is primitive,
\begin{equation}
\lim_{t \to +\infty}\max_{i}h_i(t)=\rho(C).
\end{equation}
Meantime, the power iteration \eqref{eq::pi} guarantee a compact form converge to the true eigenvector corresponding to the largest eigenvalue of adjacency matrix. In this way, 
\begin{equation}
    \lim_{t \rightarrow \infty}  \hat{y}(t+1)=\lim_{t \rightarrow \infty} \frac{C\hat{y}(t)}{\lambda_{1}(C)}=y(t)
\end{equation}
can be achieved.
\end{proof}
\begin{remark}
 In the proposed link removal algorithm, the using of Lemma \ref{lemma::perron_fr} calls for an irreducible non negative $A_r$, which equals to a strongly connected $W_r$, with the definition in \eqref{eq::def_1}. That can be a strong assumption for a matrix. However, as globalized we are these days, it is nearly impossible for any city or sub-population to stay cut off physically from the outside world. Especially in the severe pandemic situation, necessary medical and living materials must be sent by people. Therefore, no vertex is supposed to be isolated, and the assumption of a strongly connected adjacency matrix after removal is reasonable and necessary. 
\end{remark}

\section{Simulations}
\subsection{Parameter Learning via COVID-19 data in Germany}
We identified the propagation network of the COVID-19 virus consisting of 16 nodes. Each node represents a federal state in Germany. The reasons to use the data in Germany are two folds. Firstly, Germany closed the borders with neighbors on March $15$th\cite{bund.de}. Considering that the incubation period of COVID-19 is up to 14 days, the network of federal states in Germany in April and May can be viewed as isolated, i.e., infection from external nodes (other countries) is excluded. Secondly, Germany guaranteed sufficient testing capacity and numerous intensive beds, which is far below the upper limit of the healthcare resource they can provide. Therefore, the data of Germany can well interpret the infection characteristic of the virus.

We use the infection data of the 16 federal states of Germany from \cite{timmurphy.org}. The identified network is a weighted asymmetric one consisting of 16 notes. From the geographical point of view, not all the federal states are adjacent to each other. However, the network can be treated as nearly fully-connected due to the logistic, business/personal trips, etc..

Because the data from \cite{timmurphy.org} is published once a day, the propagation network is identified based on the following discrete-time model.
\begin{equation} 
    \begin{aligned}
     x[k+1]-x[k] &=(I-X[k]-R[k])BW x[k]-X[k] \delta\\
    r[k+1] -r[k] &=X[k] \delta,
    \end{aligned}
\end{equation}
which can be further simplified in linear parameterization form
\begin{equation}
    \xi[k]=\phi[k]^{\top}\theta^*
\end{equation}
with
\begin{equation}
\begin{aligned}
\phi[k]^{\top}=
 \left[\begin{array}{cc}
    x[k]^{\top} \otimes(I-X[k]-R[k]) & -X[k] \\
     0&X[k] 
\end{array}
   \right],
\end{aligned}
\end{equation}
and 
\begin{equation}
\begin{aligned}
\theta^*=\left[\begin{array}{c}
     \text{vec}(BW)\\
     d 
\end{array}
\right], \quad
\xi[k]=\left[\begin{array}{c}
     x[k+1]-x[k]\\
     r[k+1]-r[k] 
\end{array}\right],
 \end{aligned}
\end{equation}
where $\otimes$ represents the Kronecker product, $\text{vec}(BW)$ denotes the vectorization of matrix $BW$ and $I \in \mathbb{R}^{n\times n}$ is an identity matrix. We see that the network structure $BW$ and the curing rates $d$ are stacked into the parameter vector $\theta^*$. The identification of the propagation network of the virus can be formulated as a constrained optimization problem as follows
\begin{align}
\begin{split}
\label{eqn: ls}
    \theta^{*}= &\arg\min_{\theta} \frac{1}{2} \left\lVert
    \Phi^{\top}\theta - \Xi
    \right\lVert^{2}_2\\
    &\text{s.t.  } 0\leq \theta_i\leq 1, i=1,\cdots,n(n+1)
\end{split}
\end{align}
with $\Phi=[\phi[k],\phi[k+1],\cdots,\phi[k+N]]$ being the regressor matrix and $\Xi=[\xi[k],\xi[k+1],\cdots,\xi[k+N]]$ being the vector containing the data of daily increase in infected cases. $N \in \mathbb{N}$ is the number of utilized data. Since the network consists of 16 nodes, the data of at least 17 days is required to ensure the full rank of the regressor matrix. In our simulation, we utilize the data of 25 days from May 24th to April 18th, i.e., $N=25$. The optimal parameter $\theta^*$ is obtained by adopting $\mathsf{lsqlin}$ function of Matlab with interior-point algorithm. In reality, the recorded data is not perfect because of various reasons such as the delay of reporting cases and uncertain incubation periods. The identified network is the nearest solution to the real network, which satisfies the constrain in (\ref{eqn: ls}).
\iffalse
Because the original data has inestimable noise, so we need to request a local optimal solution to reduce the influence of noise on parameter identification. We only need 17 days of data to get the parameters of the model, but we need 25 days to make the parameter estimates more stable. Next we need to require the minimum value of the following expression with constraint:
\begin{equation}
    \theta^{*}=\arg\min_{\theta} \frac{1}{2} \left\lVert
\underbrace{\left[\begin{array}{c}
     \phi^{\top}[k]\\
     \phi^{\top}[k+1]\\
     \vdots\\
     \phi^{\top}[k+24]
\end{array}
\right]}_{\Phi^{\top}} \theta- \underbrace{\left[\begin{array}{c}
  \delta[k] \\
      \delta [k+1]\\
      \vdots\\
      \delta[k+24]
\end{array}
   \right]}_{\Delta}\right\lVert^{2}_2
\end{equation}
With the following constraint:
$$0\leqslant \theta_i\leqslant 1$$
We use $lsqlin$ function from Matlab to solve this problem. This function use interior-point algorithm to solve the optimal problem with constraint.
\fi

%\begin{figure}[h]
%\centering
%\includesvg[width =0.8 \linewidth]{network}
%\caption{Network Structure}
%\label{fig_network}
%\end{figure}

\begin{figure}[ht]
\centering
\includegraphics[width =1 \linewidth]{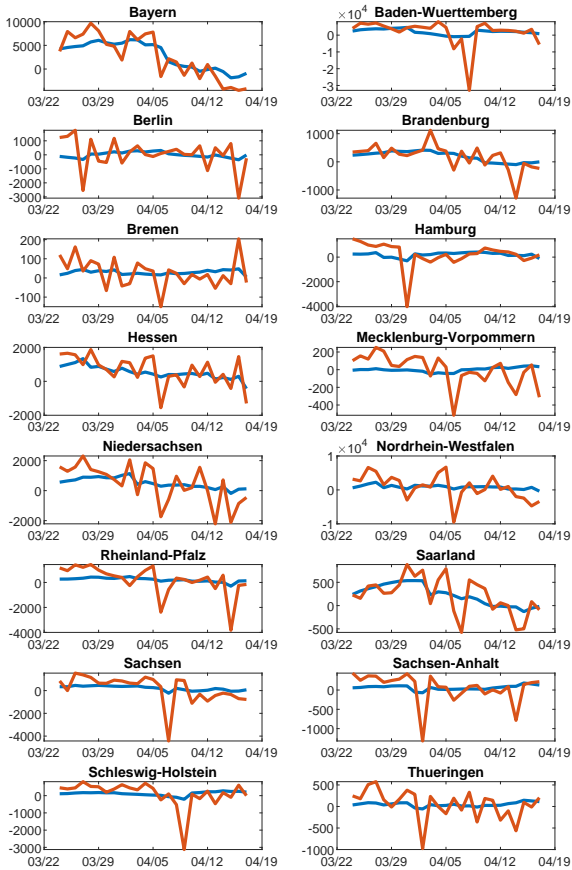}
\caption{The approximation of daily increment of infection cases in each state of Germany. The red line and blue line correspond to the real data and approximation, respectively.}
\label{fig_bayern}
\end{figure}

\subsection{Implementation of the Link Removal Strategy}
 We slightly abuse the network obtained in Section IV.A by setting the edges weight less than 0.001 to 0. This results in a strongly connected digraph and better validates our algorithm. Based on this processed graph, we initially test the estimation of the principal eigenvector, which turns out to converge well as shown in Fig. \ref{fig_vector}. The colorful $\hat{y}_i$ grows for several iterations and finally converge to the corresponding component of the principal eigenvector represented by the black dash line.

To verify the effectiveness of the proposed algorithm on solving problem \eqref{eq::opt_3}, we compare the dominant eigenvalue $\lambda_{1}(A_r)$ of the proposed removal algorithm with a random removal strategy. It worth noticing that the estimation of the dominant eigenvalue after every removal is not necessary. To remove more than one edge after an estimation is possible and can conserve computational efforts. With this concern, we set different step-size in our simulation. The Fig. \ref{fig_removal} shows that, regardless of the step-size, the proposed algorithm significantly reduces the dominant eigenvalue of $A$, so that much better performance is achieved. As for the results of different step-sizes, similar effects are achieved by 1 edge per step and 2 edges per step. Things are the same for 5 edges per step and 10 edges per step, which show little worse performance than that of 1 and 2 edges per step. In all, to minimize the step enhances the performance, but the difference is not significant if the gap between step-size settings is not huge. The links directed to Berlin and Mecklenburg-Vorpommern are more likely to be removed. That would possibly because of the high weights of links pointing to them, as the removal of fast-spreading link may helps in containing the pandemic.

 We carry an experiment in one-step scenario and further test the connectivity by abandoning the proposal that harm connectivity. In Fig. \ref{fig_removal}, the result is marked by the red cross, the line of which is only slightly different with that of blue spots representing removal without connectivity guarantee. It is possibly because the graph is well connected. Despite of the removal we have made, the connectivity is not damaged. 
 
 With the edges-removed graph, we make an estimation of the COVID-19 epidemic and compare the result with that of the approximated graph and the real data. In Fig. \ref{fig_contrast}, the decreasing rate of infection number shows obvious difference among the three lines. The real infection number drops faster than the estimation with identified parameters, which may attribute to the more experienced treatment and improving social distancing awareness. Moreover, the estimation with graph processed by the proposed algorithm shows a steeper decline than real data, which verifies that our network-based method is able to control the pandemic in an effective way.
 
\begin{figure}[ht]
\centering
\includegraphics[width =1 \linewidth]{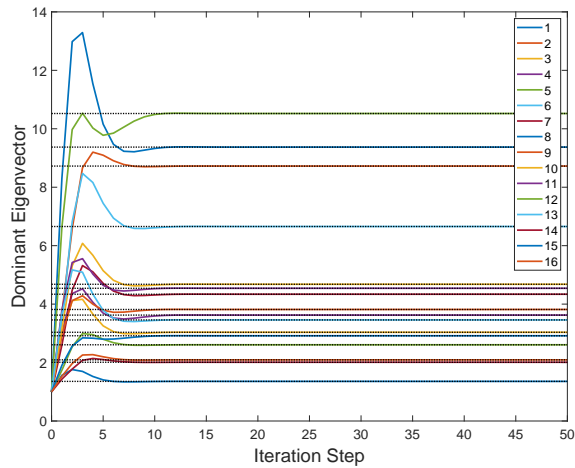}
\caption{Distributed estimation with power iteration.}
\label{fig_vector}
\end{figure}

%\begin{figure}[h]
%\centering
%\includegraphics[width =1\linewidth]{event.eps}
%\caption{Daily new cases in South Korea}
%\end{figure}

\begin{figure}[ht]
\centering
\includegraphics[width =1 \linewidth]{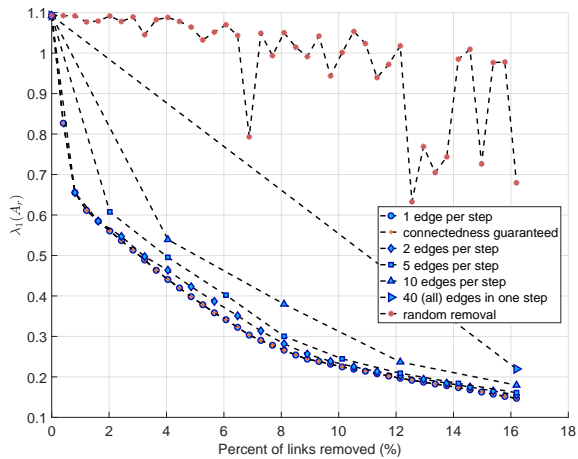}
\caption{Performance of the proposed removal algorithm with different step sizes.}
\label{fig_removal}
\end{figure}

\begin{figure}[ht]
\centering
\includegraphics[width =1 \linewidth]{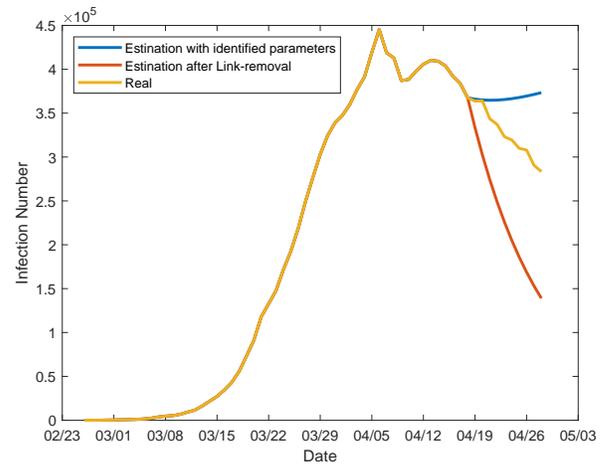}
\caption{The performance of the link-removal algorithm on the control of the COVID-19 spreading process.}
\label{fig_contrast}
\end{figure}

\section{Conclusion} \label{sec::conclusion}
In this paper, we propose a distributed link removal strategy for the network meta-population SIR epidemics. Compared with the previous work, this approach enjoys a more general setting where the investigated network is relaxed to be a weighted digraph. From practical point of view, the proposed approach is applied to the scenario of curbing the COVID-19 spreading by using infection and recovered cases in each state of Germany. The simulation illustrates the effectiveness of the proposed approach. Future work will focus on developing data-driven distributed topology manipulation strategies to control network epidemic spreading processes.

\bibliography{mybib.bib}

\begin{thebibliography}{10}
\providecommand{\url}[1]{#1}
\csname url@samestyle\endcsname
\providecommand{\newblock}{\relax}
\providecommand{\bibinfo}[2]{#2}
\providecommand{\BIBentrySTDinterwordspacing}{\spaceskip=0pt\relax}
\providecommand{\BIBentryALTinterwordstretchfactor}{4}
\providecommand{\BIBentryALTinterwordspacing}{\spaceskip=\fontdimen2\font plus
\BIBentryALTinterwordstretchfactor\fontdimen3\font minus
  \fontdimen4\font\relax}
\providecommand{\BIBforeignlanguage}[2]{{%
\expandafter\ifx\csname l@#1\endcsname\relax
\typeout{** WARNING: IEEEtran.bst: No hyphenation pattern has been}%
\typeout{** loaded for the language `#1'. Using the pattern for}%
\typeout{** the default language instead.}%
\else
\language=\csname l@#1\endcsname
\fi
#2}}
\providecommand{\BIBdecl}{\relax}
\BIBdecl

\bibitem{mei_arc_2017}
W.~Mei, S.~Mohagheghi, S.~Zampieri, and F.~Bullo, ``On the dynamics of
  deterministic epidemic propagation over networks,'' \emph{Annual Reviews in
  Control}, vol.~44, pp. 116--128, 2017.

\bibitem{nowzari_ics_2016}
C.~Nowzari, V.~M. Preciado, and G.~J. Pappas, ``Analysis and control of
  epidemics: A survey of spreading processes on complex networks,'' \emph{IEEE
  Control Systems}, vol.~36, no.~1, pp. 26--46, 2016.

\bibitem{hethcote_siam_2000}
H.~W. Hethcote, ``The mathematics of infectious diseases,'' \emph{SIAM Review},
  vol.~42, no.~4, pp. 599--653, 2000.

\bibitem{shi_medrxiv_2020}
P.~Shi, S.~Cao, and P.~Feng, ``{SEIR} transmission dynamics model of 2019 ncov
  coronavirus with considering the weak infectious ability and changes in
  latency duration,'' \emph{medRxiv}, 2020.

\bibitem{wu_lancet_2020}
J.~Wu, K.~Leung, and G.~M. Leung, ``Nowcasting and forecasting the potential
  domestic and international spread of the {2019-nCoV} outbreak originating in
  wuhan, china: a modelling study,'' \emph{The Lancet}, vol. 395, no. 10225,
  pp. 689--697, 2020.

\bibitem{liu_tcns_2020}
F.~Liu and M.~Buss, ``Optimal control for heterogeneous node-based information
  epidemics over social networks,'' \emph{IEEE Transactions on Control of
  Network Systems}, p. Online, 2020.

\bibitem{liu_physa_2019}
F.~Liu, Z.~Zhang, and M.~Buss, ``{Robust optimal control of deterministic
  information epidemics with noisy transition rates},'' \emph{Physica A:
  Statistical Mechanics and its Applications}, vol. 517, pp. 577--587, 2019.

\bibitem{nowzari_tcns_2017}
C.~Nowzari, V.~M. Preciado, and G.~J. Papas, ``Optimal resource allocation for
  control of networked epidemic models,'' \emph{IEEE Transactions on Control of
  Network Systems}, vol.~4, no.~2, pp. 159--169, 2017.

\bibitem{mieghem_ton_2009}
P.~{Van Mieghem}, J.~Omic, and R.~Kooij, ``Virus spread in networks,''
  \emph{IEEE/ACM Transactions on Networking}, vol.~17, no.~1, pp. 1--14, 2009.

\bibitem{mieghem_pre_2011}
P.~{Van Mieghem}, D.~Stevanovi\'c, F.~Kuipers, C.~Li, R.~{van de Bovenkamp},
  D.~Liu, and H.~Wang, ``Decreasing the spectral radius of a graph by link
  removals,'' \emph{Pysical Review E}, vol.~84, p. 016101, 2011.

\bibitem{xue_tsp_2019}
D.~{Xue} and S.~{Hirche}, ``Distributed topology manipulation to control
  epidemic spreading over networks,'' \emph{IEEE Transactions on Signal
  Processing}, vol.~67, no.~5, pp. 1163--1174, 2019.

\bibitem{liuj_tac_2019}
J.~Liu, P.~Par\'e, A.~Nedi\'c, C.~Y. Tang, C.~L. Beck, and T.~Basar, ``Analysis
  and control of a continuous-time bi-virus model,'' \emph{IEEE Transactions on
  Automatic Control}, vol.~64, no.~12, pp. 4891--4906, 2019.

\bibitem{varga_springer_2000}
R.~Varga, \emph{Matrix Iterative Analysis}.\hskip 1em plus 0.5em minus
  0.4em\relax Springer-Verlag, 2000.

\bibitem{Khalil_prentice_2002}
H.~K. Khalil, \emph{Nonlinear systems; 3rd ed.}\hskip 1em plus 0.5em minus
  0.4em\relax Upper Saddle River, NJ: Prentice-Hall, 2002.

\bibitem{stewart_academic_1990}
G.~Stewart and J.~Sun, \emph{Matrix Perturbation Theory}.\hskip 1em plus 0.5em
  minus 0.4em\relax Boston, MA: USA: Academic, 1990.

\bibitem{wood2003always}
R.~Wood and M.~O'Neill, ``An always convergent method for finding the spectral
  radius of an irreducible non-negative matrix,'' \emph{ANZIAM Journal},
  vol.~45, pp. 474--485, 2003.

\bibitem{bund.de}
``Coronavirus in deutschland,'' [EB/OL], \url{https://www.bundesregierung
  .de/breg-de/themen/coronavirus/kontrollen-an-den-grenzen-1730742}, Accessed
  May 25, 2020.

\bibitem{timmurphy.org}
``Coronavirus-monitor,'' [EB/OL],
  \url{https://interaktiv.morgenpost.de/corona-virus-karte-infektionen-deutschland-weltweit/},
  Accessed May 25, 2020.

\end{thebibliography}
\bibliographystyle{IEEEtran}

%\appendix \label{sec::appendix}



%\vfill



\end{document}